\documentclass{llncs}
\usepackage{amsmath,amssymb}
\usepackage{graphicx,hyperref,cite}
\usepackage{microtype}

\newif\iffull\fulltrue

\let\doendproof\endproof
\renewcommand\endproof{~\hfill\qed\doendproof}

\newcommand{\R}{\mathbb{R}}

\title{Realization and Connectivity of the Graphs of Origami Flat Foldings}
\author{David Eppstein}
\institute{Department of Computer Science, University of California, Irvine\thanks{Supported in part by NSF grants CCF-1618301 and CCF-1616248.}}

\begin{document}
\maketitle

\begin{abstract}
We investigate the graphs formed from the vertices and creases of an origami pattern that can be folded flat along all of its creases. As we show, this is possible for a tree if and only if the internal vertices of the tree all have even degree greater than two. However, we prove that (for unbounded sheets of paper, with a vertex at infinity representing a shared endpoint of all creased rays) the graph of a folding pattern must be 2-vertex-connected and 4-edge-connected.
\end{abstract}

\section{Introduction}

This work concerns the following question: Which graphs can be drawn as the graphs of origami flat folding patterns?

In origami and other forms of paper folding, a \emph{flat folding} is a type of construction in which an initially-flat piece of paper is folded so that the resulting folded shape lies flat in a plane and has a desired shape or visible pattern. This style of folding may be used as the initial base from which a three-dimensional origami figure is modeled, or it may be an end on its own. Flat foldings have been extensively studied in research on the mathematics of paper folding.
The folding patterns that can fold flat with only a single vertex have been completely characterized,  for standard models of origami~\cite{Jus-BO-86,Hul-SEC-94,Huf-TC-76,HusHus-TGO-79,Rob-PRSE-77,Kaw-OST-89,Mur-BJCA-66a,Mur-BJCA-66b}, for \emph{rigid origami} in which the paper must continuously move from its unfolded state to its folded state without bending anywhere except at its given creases~\cite{AbeCanDem-JoCG-16}, and even for single-vertex folding patterns whose paper does not form a single flat sheet~\cite{AbeDemDem-JoCG-18}. However, the combinatorics of multi-vertex flat folding patterns is much less well understood, and testing whether a multi-vertex pattern folds flat is NP-hard~\cite{BerHay-SODA-96}.

From the point of view of graph drawing, origami folding patterns can be thought of as planar graphs, drawn with straight line edges in the Euclidean plane, with each edge representing a crease that must be folded. For instance, the familiar bird base, a starting point for the classic three-dimensional origami crane, can be thought of as a graph drawing of a planar graph with 13 vertices (\autoref{fig:bird-base}). This naturally raises the question (analogous to similar questions for other types of geometric graphs such as Voronoi diagrams~\cite{LioMei-CGTA-03}): which graphs can be drawn this way?
The NP-completeness of recognizing multi-vertex flat folding patterns does not extend to this question, because the completeness result is for folding patterns that have already been embedded with a given geometry and its proof depends on the specific geometry of the embedding. Here, instead, we ask whether an embedding exists.
We do not resolve this question, but we provide partial answers to it in two different directions.

\begin{figure}[t]
\centering\includegraphics[width=0.6\textwidth]{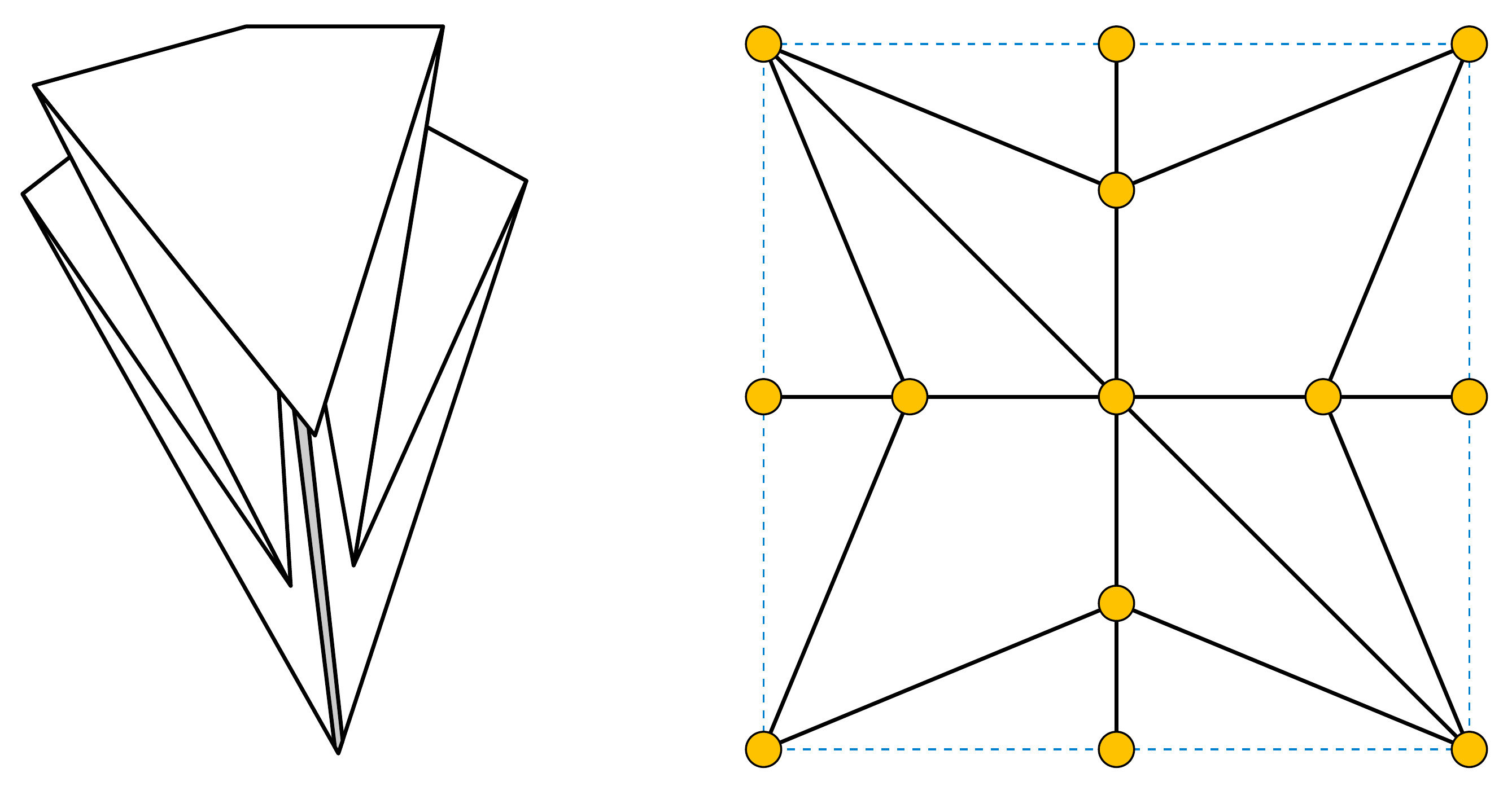}
\caption{Origami bird base (as illustrated by Fred the Oyster at \url{https://commons.wikimedia.org/wiki/File:Bird_base.svg}) and the corresponding folding pattern, interpreted as a graph drawing. The black lines indicate the final creases of the bird base. Temporary creases made while folding the base but later flattened out are not included. Blue dashed lines indicate the boundary of the sheet of paper; these lines are not considered as edges of the graph because they are not creased.}
\label{fig:bird-base}
\end{figure}

First, we investigate the trees that may be drawn as flat folding patterns. For this problem, we make the simplifying assumption that the sheet of paper to be folded is infinite, with internal vertices of the tree at points where multiple creases come together, and with the leaves of the tree corresponding to creases along infinite rays. Cutting the infinite paper of such a drawing along a square that surrounds all the internal vertices would produce a finite representation of the same tree with its leaves on the boundary of the square, like the representation of a non-tree graph in \autoref{fig:bird-base}.
Similar tree-drawing styles, with infinite rays for the leaves of the trees, have been used in past work on drawings of trees as Voronoi diagrams~\cite{LioMei-CGTA-03}, straight skeletons~\cite{AicCheDev-CCCG-12},\footnote{Straight skeletons have also been used to construct folding patterns~\cite{DemDemLub-JCGCG-98}. However, this technique adds extra folds to the skeleton, so the realizations of trees as straight skeletons do not yield realizations of the same trees as flat folding patterns.} or with optimized angular resolution~\cite{CarEpp-GD-06}.
For this model of origami folding and tree realization, we provide a complete characterization: a tree may be drawn in this way if and only if all of its internal vertices have even degree greater than two.

Second, we investigate the connectivity restrictions on the graphs that may be drawn as flat folding patterns. This type of constraint has proven very fruitful in past questions about the geometric realizations of planar graphs, providing complete characterizations of the graphs of convex polyhedra (Steinitz's theorem)~\cite{Ste-EMW-22}, drawings with rectangular faces (``rectangular duals'')~\cite{KozKin-Nw-85,BhaSah-Algo-88,He-SICOMP-93,KanHe-WG-93}, orthogonal polyhedra~\cite{EppMum-JoCG-14}, and two-dimensional soap bubble clusters~\cite{Epp-DCG-14}.

Trees are not highly connected, and may be drawn as flat folding diagrams, but it turns out that that these diagrams remain highly connected through the boundary of the drawing. To capture this boundary connectivity, we modify our mathematical model of flat folding. We again assume an infinite sheet of paper, but we treat creases along infinite rays as all having a single shared endpoint at infinity, which forms another vertex of the graph. In this model, the tree foldings of the other model become series-parallel graphs, in which all the leaves of the tree have been merged into a single supervertex.

We prove that, for this model of graphs as folding patterns, the graphs that may be realized are highly restricted, beyond even the graphs of polyhedra and beyond the immediate restriction (from the one-vertex case) that all vertices have even degree. In particular, they are necessarily 2-vertex-connected and 4-edge-connected. More strongly, the vertex at infinity is not an articulation vertex, and any subset of vertices that separates the graph and does not include the vertex at infinity must include at least four other vertices. These connectivity restrictions hold even for a weaker model of \emph{local flat foldability} in which we seek a piecewise linear map from the folding pattern to its folded state in the plane without regard to whether this folding can be embedded without self-intersections into three-dimensional space. Our realizations of trees as flat folding patterns show that the 2-vertex-connectivity and 4-edge-connectivity conditions are both tight: no higher restriction on connectivity is possible.

\section{Preliminaries}

\subsection{Mathematical model of folding}

Departing from the usual square-paper model of origami in order to avoid complications from its boundary conditions, we model the sheet of paper to be folded as the entire Euclidean plane.
We first define a \emph{local flat folding}. This is a highly simplified model of how a piece of paper might be folded that only takes into account local constraints (the paper can only be folded, not stretched, sheared, or crumpled), does not prevent self-intersections, and does not even represent the most basic information about how the folding might occur in three dimensions, such as whether a given fold is a mountain fold or a valley fold. 

\begin{definition}
We define a continuous function $\varphi$ from the plane to itself to be a local flat folding if every point $p$ of the plane has one of the following three types:
\begin{itemize}
\item An \emph{unfolded point} of a local flat folding is a point $p$ such that $\varphi$ is a \emph{local isometry}: there is a neighborhood of $p$ that is mapped by $\varphi$ in a distance-preserving way (necessarily a combination of translation, rotation, or reflection of the plane).
\item A \emph{crease point} of a local flat folding is a point $p$ that has a neighborhood $N$
that can be covered by two subsets, each containing $p$ and each mapped by $\varphi$ in a different distance-preserving way. Necessarily, the boundary between these two subsets must be a line containing $p$. To preserve continuity of the mapping, the two distinct isometric mappings for the two subsets must be reflections of each other across the image of this line. The points within $N$ that belong to this fold line are also crease points, and the other points within $N$ are unfolded points.
\item A \emph{vertex point} of a local flat folding is a point $p$ that has a neighborhood $N$ that can be covered by finitely many (and at least three) subsets, each containing $p$ and each mapped by $\varphi$ in a distance-preserving way so that there are at least three distinct isometric mappings among these subsets. Necessarily, each subset must be a wedge. The points within $N$ that belong to the rays between pairs of wedges are crease points, and the points within $N$ that do not belong to these rays are unfolded points.
\end{itemize}
Then, as stated above, a local flat folding is a continuous function $\phi$ such that all points of the plane are unfolded points, crease points, and vertex points. We add one more restriction: we consider only local flat foldings that have at least one vertex point. We do not require the number of vertex points to be finite.
\end{definition}

As a simple example, consider the function $\varphi: (x,y)\mapsto (f(x),f(y))$
where $f(x)=|(x\operatorname{mod}2)-1|$.
Here $f$ is a continuous function that maps the intervals $[2i,2i+1]$ to $[0,1]$ in reverse order,
and that maps the intervals $[2i+1,2i+2]$ to $[0,1]$ linearly.
$\varphi$ corresponds to a folding pattern in which we \emph{pleat} the plane along the integer-coordinate vertical lines (that is, we create a sequence of folds that alternates between mountain and valley folds, like an accordion; see \cite[p.~31]{Lan-ODS-12}), and then we pleat it again along the integer-coordinate horizontal lines, so that the whole plane is mapped to the unit square. Its folding pattern has vertex points at points of the plane where both coordinates are integers, crease points at points with one integer coordinate, and unfolded points everywhere else. That is, it is a drawing of the infinite square grid graph.

In general, the graph of a local flat folding is almost a graph drawing, in that its vertex points form a discrete set, connected in pairs by line segments consisting of crease points. For the grid example, it is a graph drawing. However, for other local flat foldings, some of the crease points may belong to semi-infinite rays rather than forming bounded line segments. To make a graph that also includes these rays as edges, we add a special vertex $\infty$ that is not represented by any geometric point, and we treat this special vertex as an endpoint of each ray of crease points.

\begin{definition}
We define the \emph{graph of a local flat folding} $\varphi$ to be a graph $G$ that has a vertex for each vertex point of $\varphi$ and (if $\varphi$ includes any infinite rays of crease points) another special vertex $\infty$. Two vertex points form adjacent vertices in $G$ when the line segment between them consists only of crease points. A vertex point $p$ and the special vertex $\infty$ are adjacent when there exists a ray with apex $p$ consisting only (other than at its apex) of crease points. This graph may have multiple adjacencies between $\infty$ and other vertices (for instance, it will do so in any one-vertex flat folding pattern) but it can have at most one edge between any two vertex points.
\end{definition}

The folding pattern provides a topological planar embedding for the whole graph $G$, and a geometric straight-line planar embedding for all vertices except $\infty$. As usual, we call the maximal regions of the plane that are disjoint from the vertices and edges of the embedding (the vertex and crease points of $\varphi$) the \emph{faces} of the embedding. These are possibly-unbounded polygonal regions, the connected components of the unfolded points of $\varphi$.
Because the action of $\varphi$ on each face of the graph is determined from its action on adjacent faces, the embedding of $G$ completely determines the mapping of $\varphi$, up to a congruence transformation of the whole plane.

For our realizations of trees, we will use a slightly different graph, that can be derived from the graph of the folding. (It will not be interesting to study the graph connectivity of this graph, because it will have many degree-one vertices.)

\begin{definition}
We define the \emph{truncated graph of a local flat folding} to be the graph formed in either of the following two equivalent ways:
\begin{itemize}
\item From the graph of the folding, subdivide each edge incident to $\infty$, and then delete vertex $\infty$.
\item Form a graph with a vertex for each vertex point of the folding and another vertex for each ray of crease points of the folding. Connect two vertex points by an edge if the line segment between them consists only of crease points. Add an edge for each ray of crease points, connecting the vertex point at the apex of the ray to the additional vertex for the same ray.
\end{itemize}
\end{definition}

Truncated graphs of local flat foldings can also be interpreted as the type of graph drawn in \autoref{fig:bird-base} for a folding pattern on a sheet of square paper with the additional property that the creases reaching the boundary form diverging rays. However, the folding pattern in \autoref{fig:bird-base} has creases that instead meet at the boundary, and it is also possible to form converging pairs of rays. Therefore the type of graph shown in the figure, of a folding pattern on a bounded square of paper, is somewhat more general. However, for the purposes for which we use truncated graphs (realization of trees), a less general model is better, as any realization in such a model will also be a realization for the more general model.

It remains to define a mathematical model of foldings as global structures, accounting for how paper can fold in three dimensions and how some parts of the paper can block other parts of paper from passing through them (disallowing self-intersections). It is possible to model precisely the above-below relation of the faces of $\varphi$, and the nesting structure of the folding at the creases of $\varphi$; see, for instance,~\cite{AbeDemDem-JoCG-18} for a similar model of lower-dimensional flat-folded structures.
However, we will forgo the complexity of such a model in favor of the following simpler topological approach.

\begin{definition}
A \emph{global flat folding} is a local flat folding $\varphi$ with the additional property that,
for every $\epsilon>0$, there exists a topological embedding $\varphi_\epsilon:\R^2\to\R^3$
(without self-intersections) such that composing $\varphi_\epsilon$ with the coordinatewise vertical projection from $\R^3$ to $\R^2$ produces a mapping that, for every point $p$, is within distance $\epsilon$ of the mapping given by $\varphi$.
\end{definition}

Intuitively, a global flat folding is a local flat-folding that, for every $\epsilon>0$, is $\epsilon$-close to a topological embedding of the plane into three-dimensional space.

\subsection{Single-vertex restrictions}

The geometry of single-vertex folding patterns, such as the one in \autoref{fig:one-vertex}, is characterized by Maekawa's theorem and Kawasaki's theorem~\cite{Jus-BO-86,Hul-SEC-94,Huf-TC-76,HusHus-TGO-79,Rob-PRSE-77,Kaw-OST-89,Mur-BJCA-66a,Mur-BJCA-66b}. These apply as well to each vertex of a multi-vertex folding pattern.

\begin{theorem}[Maekawa's theorem for single-vertex folding patterns without mountain-valley assignments]
Each vertex point of a folding pattern must be incident to an even number of creases.
\end{theorem}

This follows easily from the observation that, at each crease, the paper alternates between having its top side up (a region within which $\varphi$ is an orientation-preserving isometric mapping) and having its bottom side up (a region within which $\varphi$ is an orientation-reversing isometric mapping).

\begin{figure}[t]
\centering\includegraphics[width=0.5\textwidth]{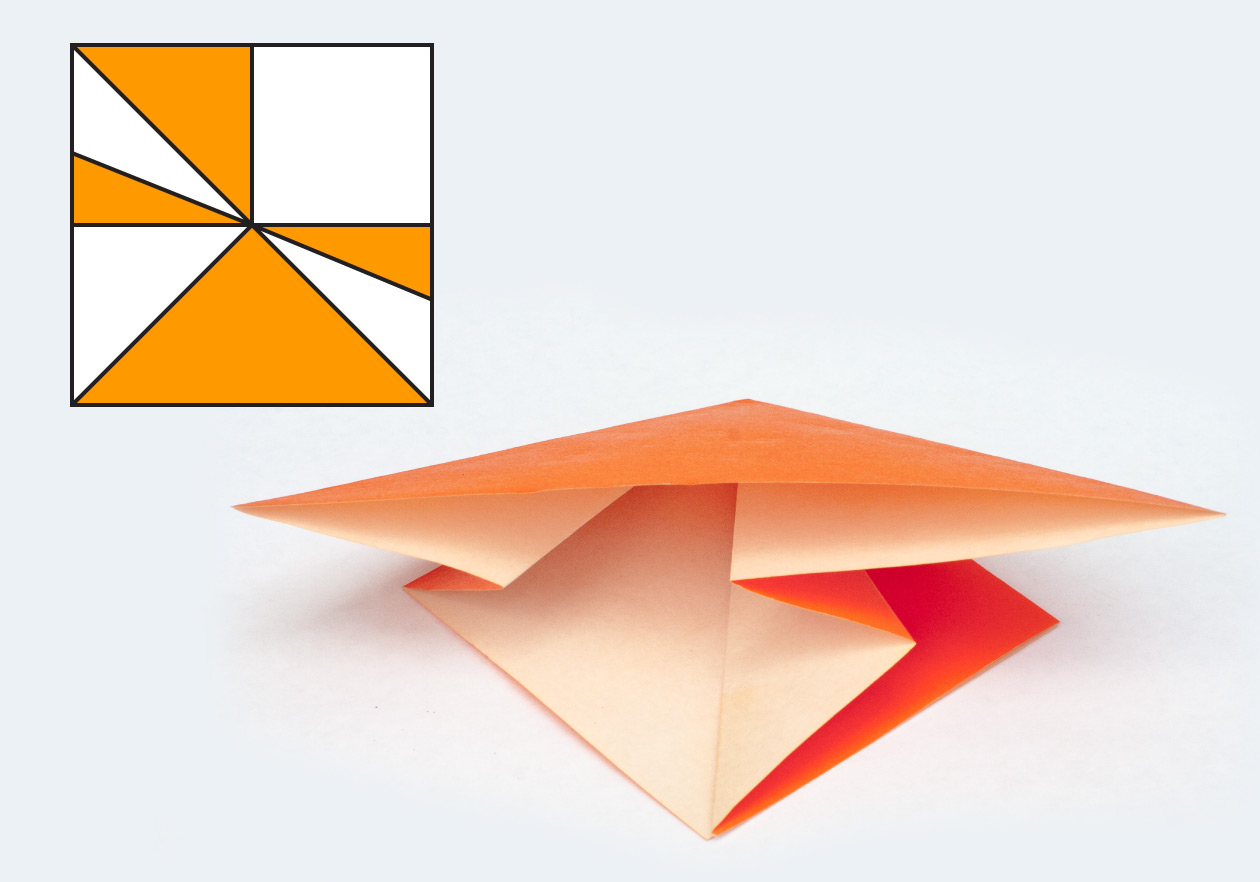}
\caption{A single-vertex flat folding and its pattern, demonstrating Maekawa's theorem (the number of folds is even) and Kawasaki's theorem (the face-up orange total angle equals the bottom-up white total angle). Image by the author for Wikipedia, 2011.}
\label{fig:one-vertex}
\end{figure}

\begin{theorem}[Kawasaki's theorem]
At each vertex point of a folding pattern, the alternating sum of wedge angles totals to zero.
\end{theorem}

This again follows from the fact that, near the vertex in the flat-folded state of the pattern,
each point is covered by equal numbers of upward-facing and downward-facing regions, so
the total amount of upward-facing paper must equal the amount of downward-facing paper.

\begin{corollary}
Each wedge of a vertex point of a flat folding has angle strictly less than $\pi$. Therefore, each face of a flat folding pattern is a (possibly unbounded) convex polygon.
\end{corollary}

\section{Realization of trees}
\label{sec:trees}

Let $T$ be any plane tree. Then by Maekawa's theorem, if $T$ is to be realized as the truncated graph of a local flat folding, its internal vertices must have even degree greater than two. Our purpose in this section is to prove that this condition is necessary as well as sufficient.

\begin{figure}[t]
\centering\includegraphics[width=0.4\textwidth]{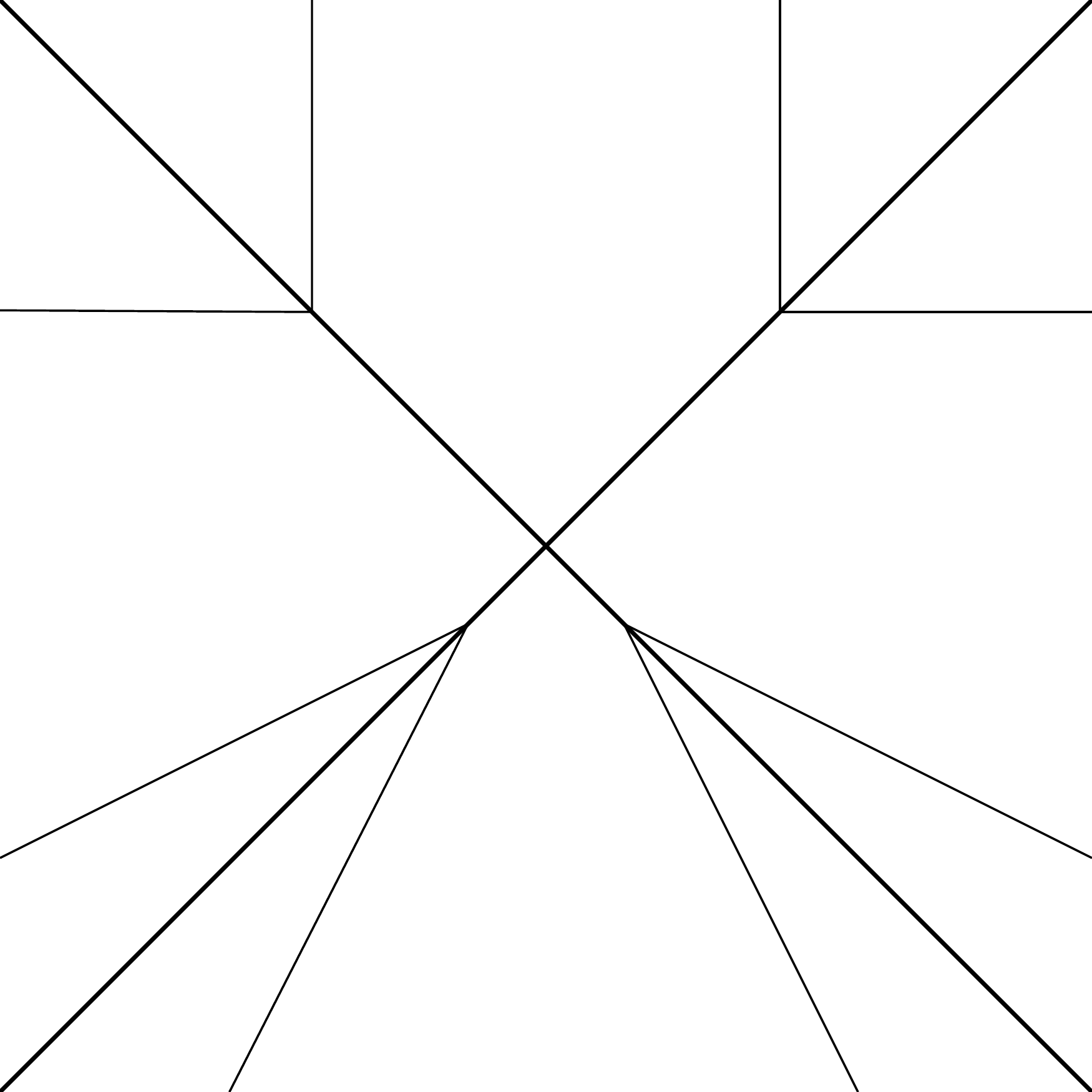}
\caption{A tree folding pattern that can be locally flat folded, but not globally flat folded.}
\label{fig:tree-counterexample}
\end{figure}

We are interested here in global flat foldings, not just local flat foldings, and for this reason some care must be taken. It is not sufficient merely to embed $T$ as a graph in the plane, with its leaf edges drawn as rays, and with each internal vertex meeting the angle sum condition of Kawasaki's theorem. \autoref{fig:tree-counterexample} depicts a counterexample.
It obeys Kawasaki's theorem, and can be locally flat folded, but not globally flat folded.
The four heavier diagonal lines of the figure can be flat folded in only one way up to combinatorial equivalence. Their folding is obtained by first folding along one diagonal line, and then along the other. The four creases of this fold are then modified by subsidiary folds that are each individually possible. But one of the four heavier creases must be nested tightly within another one. The two subsidiary creases of these two nested creases are arranged in such a way that, no matter which crease is nested within the other, the subsidiary crease of one will be blocked by the paper from the other nested crease. (Try it!)

To evade this problem, we seek a stronger type of realization, one in which each crease is ``protected'' by a wedge surrounding it, within which we can add modifications (such as the subsidiary wedges of \autoref{fig:tree-counterexample}) without interfering with other parts of the folding.

\begin{theorem}
Let $T$ be any finite tree with all internal vertices having even degree greater than two. Then $T$ can be realized as the truncated graph of a global flat folding.
\end{theorem}

\begin{proof}
We use induction on the number of internal nodes of $T$ to prove a stronger statement: that $T$ can be realized in such a way that each ray $r$ of $T$ is associated with a wedge $W_r$, satisfying the following properties:
\begin{itemize}
\item Ray $r$ and wedge $W_r$ have the same apex, and $r$ is the median ray of its wedge.
\item Each two rays have interior-disjoint wedges. Each edge of $T$ that is not a ray is disjoint from all of the wedges.
\item There exists a three-dimensional folded state such that the two halves of each wedge $W_r$ are placed touching each other, with no other paper between them.
\end{itemize}
The third property above is phrased informally, so let us relate it to our earlier topological definition of a global flat folding. Recall that, in order to formalize the notion of a ``three-dimensional folded state'' we really have a parameterized family of three-dimensional embeddings. That is, we have both a folding map $\varphi:\R^2\to \R^2$ and, for each $\epsilon>0$, a topological embedding $\varphi_\epsilon:\R^2\to \R^3$ whose vertical projection to $\R^2$ is $\epsilon$-close to $\varphi$. We formalize the ``no other paper between them'' constraint, again up to $\epsilon$-closeness: for each point $p\in\R^2$ at a distance of $\epsilon$ or more from the boundary of $\varphi(W_r)$,
the preimage of $p$ (according to the vertical projection) in $\varphi_\epsilon(\R^2)$ should have two points from the two sides of $W_r$ consecutive with each other in the vertical ordering of the points.

\begin{figure}[t]
\centering\includegraphics[width=0.8\textwidth]{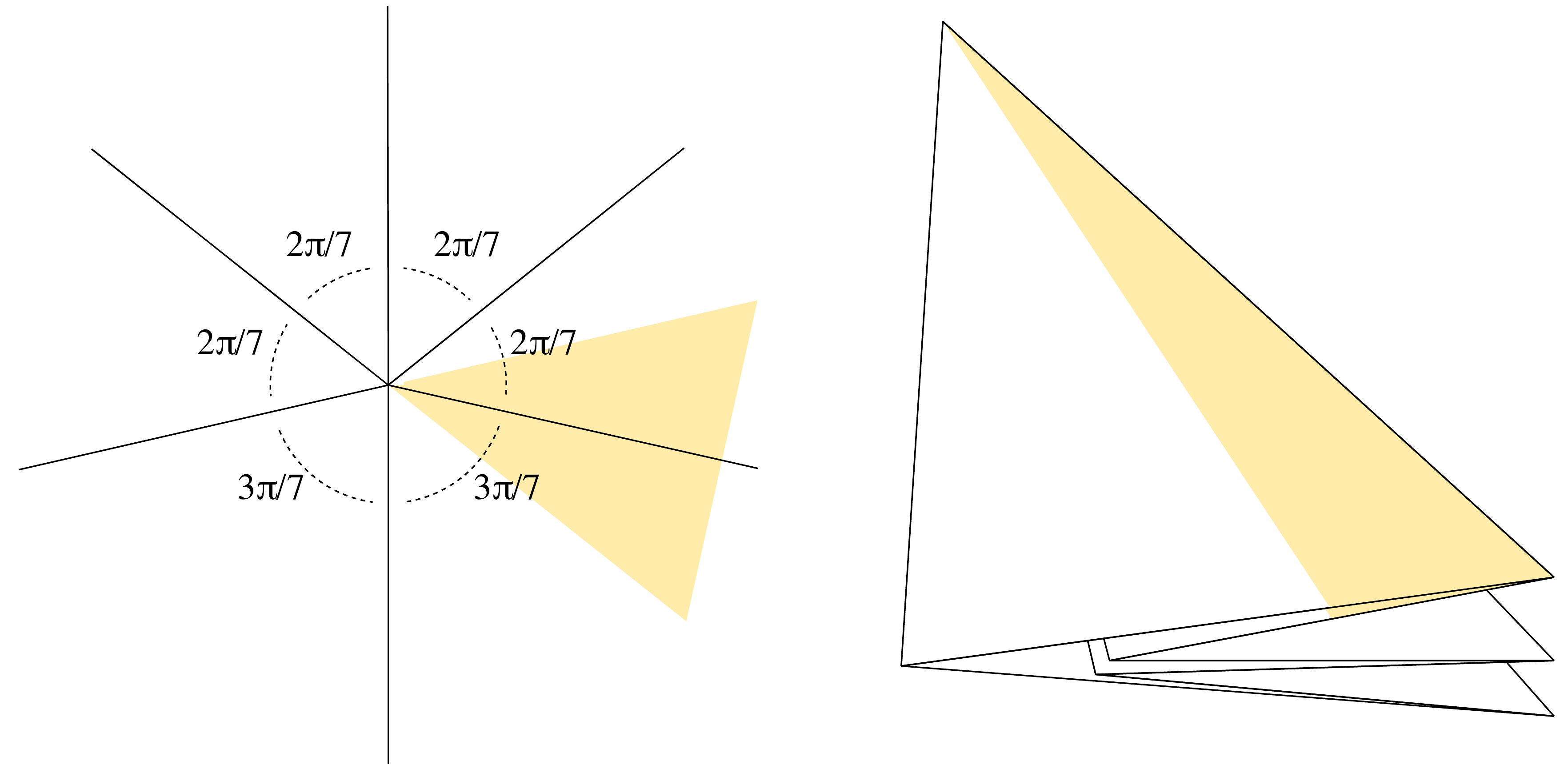}
\caption{The base case for realizing a one-internal-vertex tree (here with degree $d=6$), showing the wedge $W_r$ for one of the rays $r$ both in the folding pattern and in the folded state.}
\label{fig:tree-base}
\end{figure}

The base case of the induction is a tree $T$ with one internal node $v$ of even degree $d$ greater than four.
In this case, we let $\theta=\pi/(d+1)$. We draw $T$ as a set of $d$ rays, all meeting at a common point. We make two of the angles between consecutive rays of $T$ equal to $3\theta$, and all remaining angles equal to $2\theta$. For instance, when $d=7$, we get $\theta=\pi/7$ and six rays separated by angles of $3\pi/7,3\pi/7,2\pi/7,2\pi/7,2\pi/7,2\pi/7$.
We fold this in three dimensions by placing the two wider wedges on the top and bottom of the folded pattern, and pleating the remaining wedges between them.  For this fold, we make each wedge $W_r$ for a ray $r$ of the folding pattern be the wedge centered on that ray with opening angle $2\theta$. This opening angle is sufficient to make all the wedges interior-disjoint, and it is straightforward to verify that the 3d realization of this fold places no paper between the two halves of any wedge. This case is depicted in \autoref{fig:tree-base}.

Otherwise, if $T$ has more than one internal vertex, let $v$ be any internal vertex that has only a single non-leaf neighbor. (For instance, $v$ may be found by choosing any vertex $u$ arbitrarily and letting $v$ be an internal vertex that is maximally far from $u$.)
Let $T'$ be the tree formed from $T$ by removing the leaf neighbors of $v$, so that $v$ itself becomes a leaf. Then by the induction hypothesis, $T'$ can be realized by a global flat folding, with a ray $r$ that is associated with its leaf $v$ and that is surrounded by a wedge $W_r$, whose two halves touch each other without being blocked by other paper in the folding. Let $\theta$ denote the opening angle of wedge $W_r$. Suppose also that, in $T$, $v$ has degree $d$, and therefore it also has $d-1$ leaf children.

Then we modify the folding that represents $T'$ to form a folding representing $T$, as follows.
We place $v$ at an arbitrarily chosen point along $r$ (for instance, at the point a unit distance away from the apex of ray $r$). Then, we form $d-1$ creases, along $d-1$ rays  with $v$ as apex, to represent the $d-1$ leaf children of $r$.  We choose the angles of these rays so that they are separated from each other and from the two boundary rays of $W_r$ by an angle of $\theta/d$.
Finally, we assign each of these rays its own wedge, with $v$ as its apex and with opening angle $\theta/d$. (See \autoref{fig:add-to-tree}.)

\begin{figure}[t]
\centering\includegraphics[width=0.5\textwidth]{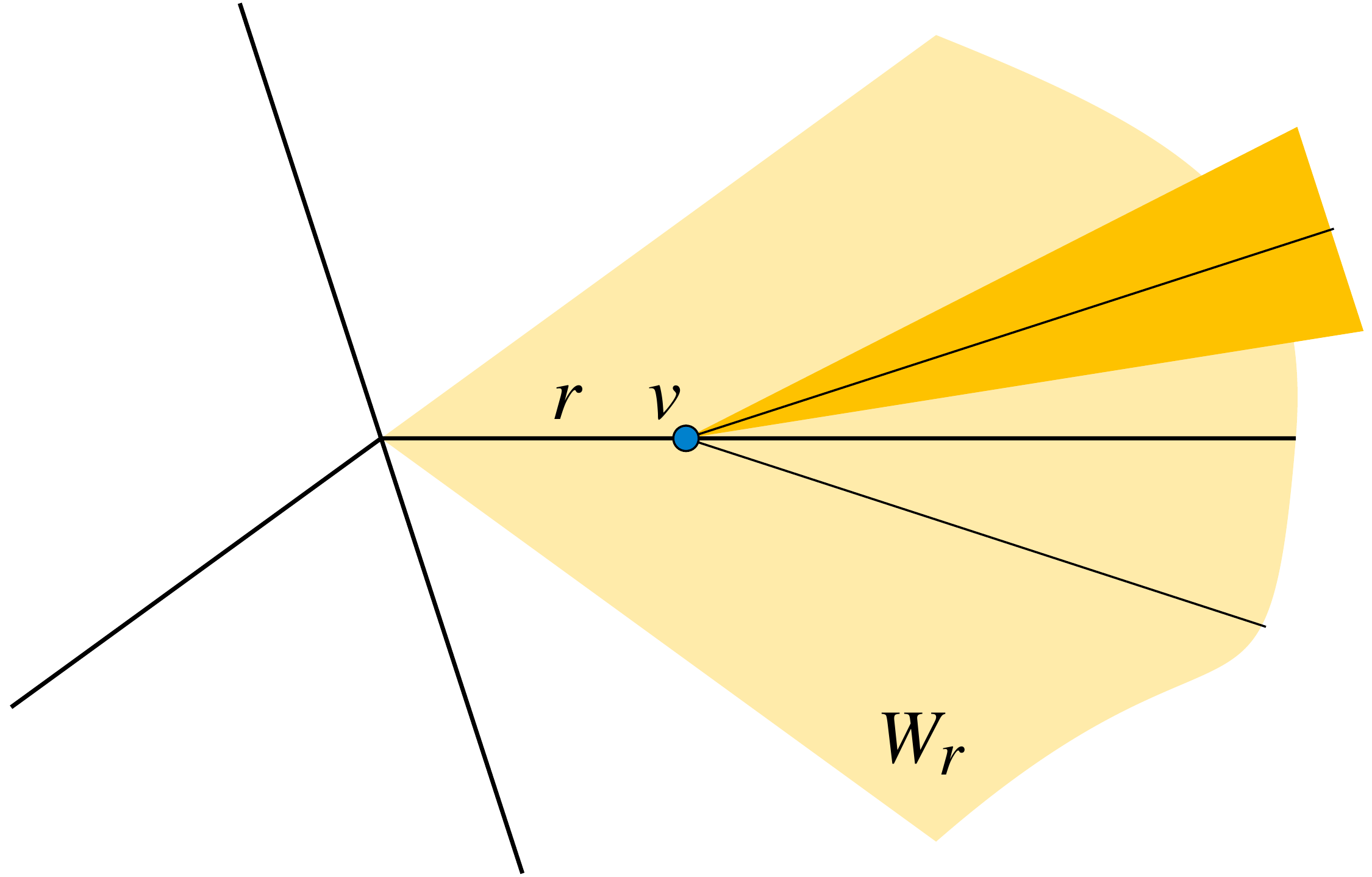}
\caption{Adding a vertex $v$ to the folding of $T'$ to create a folding for $T$. We choose the angles of the new rays incident to $v$ so that they and the two boundary rays of the outer wedge $W_r$ are equally spaced. The wedge surrounding each new ray has opening angle equal to the spacing of the rays. The crease pattern of the figure corresponds to a tree with two degree-four internal nodes.}
\label{fig:add-to-tree}
\end{figure}

The 3d folding of the crease pattern for $T'$ can also be modified in the same way to form a 3d folding for the crease pattern for $T$. At $v$, the rays and segments representing incident edges of $T$ form $d$ wedges, two of which have opening angle greater than $\pi$ and the rest of which have opening angle $\theta/d$. As before, we fold this part of the paper so that the two large wedges are outermost and the other wedges are pleated between them. The angles of the creased rays are chosen so that, after this pleat, the creases that are folded to become the closest to the boundary rays of $W_r$ (such as the middle ray of the figure) become parallel to these boundary rays.
Because of this, the folded state stays within the region of $\R^3$ previously occupied by the paper for wedge $W_r$, and the empty space between the two sides of that wedge, so it does not interfere with any other part of the global flat folding. Each of the wedges of opening angle $\theta/d$ surrounding the new rays of the folding has its two sides mapped directly above and below each other in the pleating, maintaining the invariant of the induction.
\end{proof}

We remark that, because the pleating pattern used for this realization does not ever tightly nest one crease inside another, it is possible to find a 3d realization that projects exactly to the two-dimensional local flat folding, rather than approaching it through $\epsilon$-approximations.

\section{Connectivity}

Although we have seen that truncated graphs of flat foldings may be trees (graphs that are not very highly connected), we now show that the full graph, including the special vertex $\infty$, is (when finite) always well connected. We assume throughout this section that the full graph has at least one finite vertex; otherwise, as a one-vertex graph, the full graph is trivially $k$-vertex-connected and $k$-edge-connected for all $k$.

\begin{lemma}
\label{lem:not-articulation}
Let $G$ be the graph of a local flat folding.
Then the special vertex $\infty$ is not an articulation vertex of $G$.
\end{lemma}

\begin{proof}
If it were, some two components of $G-\infty$ would necessarily be separated by an infinite face of the folding pattern.
However, because all faces are convex each connected component of the boundary of an infinite face forms a convex polygonal chain, ending in two rays that span an angle (within the face) of less than $\pi$ with each other. It is not possible for two such chains to bound a single face without crossing each other, so the boundary of the face can have only one connected component.
\end{proof}

\begin{lemma}
\label{lem:nearby-rigid}
Let $u$ and $v$ be two vertex points of a local flat folding $\varphi$ that belong to the same face of $\varphi$ and let $d$ denote Euclidean distance.
Then $d(u,v)=d(\varphi(u),\varphi(v))$.
\end{lemma}

\begin{proof}
Because the faces of $\varphi$ are strictly convex, the line segment between $u$ and $v$ must either consist entirely of crease points (on an edge of the graph of the folding) or unfolded points (if $u$ and $v$ are not consecutive on their shared face). In either case this line segment is mapped to an equal-length line segment by $\varphi$.
\end{proof}

\begin{lemma}
\label{lem:no-3-separation}
Let $G$ be the finite graph of a local flat folding. Then removing up to three of the vertex points of the folding from $G$ cannot cause the remaining graph to become disconnected.
\end{lemma}

\begin{proof}
Suppose for a contradiction that $S$ is a set of at most three vertex points whose removal disconnects $G$.
Since $G$ is a plane graph, there must exist a simple closed curve $C$ in the plane that passes through $S$ 
and is otherwise disjoint from the vertices and edges of $G$, with at least one vertex inside the curve and at least one vertex outside the curve. (For folding patterns that include a ray of crease points, we count $\infty$ as being outside all such curves.) But as we show in the case analysis below, this is not possible:
\begin{itemize}
\item If $|S|=1$, any curve $C$ through the single vertex of $S$ that is otherwise disjoint from $G$ must remain within a single convex face of $G$, and cannot enclose anything.
\item If $S$ consists of two non-adjacent vertices, they can only have one face of $G$ in common.
Any curve $C$ through these two vertices that is otherwise disjoint from $G$ must remain within that face, and cannot enclose anything.
\item If $S$ consists of two adjacent vertices, then a curve $C$ through the two vertices $u$ and $v$ of $S$ that  is otherwise disjoint from $G$ can either stay within one of the two faces incident to edge $uv$ (not enclosing anything) or have one arc in one of these two faces and one arc in the other of the two faces, enclosing edge $uv$ but not enclosing any vertices.
\item If $S$ consists of three collinear vertex points, then curve $C$ must visit each of these three points in turn.
But the outermost of these two vertex points cannot belong to any convex face of the folding pattern (because this face would also contain the middle point), and cannot be connected by an arc of $C$.
\item If $S$ consists of three non-collinear vertex points $u$, $v$, and $w$, then $C$ can only enclose any vertex points that might lie interior to triangle $uvw$. However, triangle $uvw$ is mapped by the local flat folding map $\varphi$ to a congruent triangle, by \autoref{lem:nearby-rigid} and by the fact that there is only one Euclidean triangle (up to congruence) for any triple of distances between its vertices. In order to avoid stretching, every line segment formed by intersecting a line with triangle $uvw$ must be mapped by $\varphi$ to the corresponding line segment of the image triangle. In particular, there can be no creases within triangle $uvw$, because whenever a line segment properly crosses a crease of a local flat folding, it is not mapped to a congruent line segment. Therefore, every point inside triangle $uvw$ must be an unfolded point, and $C$ cannot contain a vertex point.
\end{itemize}
Because there is no way to construct curve $C$, the hypothesized set $S$ cannot exist.
\end{proof}

The assumption that $G$ is finite is used in the existence of $C$. If $G$ could be infinite, our tree realization construction could be used to construct a realization of an infinite tree in which $\infty$ is a degree-one leaf. This does not have the connectivity described by the lemma, but this is not a contradiction because it does not meet the assumptions of the lemma. 

\begin{theorem}
If $G$ is the finite graph of a local flat folding $\varphi$, then $G$ is 2-vertex-connected and 4-edge-connected.
\end{theorem}

\begin{proof}
$G$ can have no articulation vertex, because neither $\infty$ nor any vertex point of $\varphi$ can be an articulation vertex (\autoref{lem:not-articulation} and \autoref{lem:no-3-separation} respectively).

Assume for a contradiction that $G$ could have three edges $e_1$, $e_2$, and $e_3$ whose removal disconnects $G$. Choose a vertex point $v_i$ as one of the two endpoints of each of these edges (as each edge in $G$ has at least one vertex point as its endpoint). The separation of $G$ caused by the removal of the edges $e_i$ cannot separate any subset of the three vertices $v_i$ from the rest of $G$, because $G$ has minimum degree four and, in a graph of this degree, any set of up to three vertices is connected to the rest of the graph by at least four incident edges.
Therefore, there must be at least one vertex of $G$ on each side of the separation that is not one of the three chosen vertices $v_i$. However, this implies that these three vertices also separate $G$, contradicting \autoref{lem:no-3-separation}. This contradiction implies that our assumption is false, and therefore that $G$ is 4-edge-connected.
\end{proof}

We remark that our realizations of 4-regular trees show that both 2-vertex-connectivity and 4-edge-connectivity are tight: some graphs that can be realized as global flat foldings are neither 3-vertex-connected nor 5-edge-connected.

\section{Conclusions}

We have shown that trees can be realized as the (truncated) graphs of flat folding patterns, and that despite this the (non-truncated) graphs of flat folding patterns must be highly connected.
However we have not succeeded in completely characterizing the graphs of flat folding patterns.
We leave the following questions as open for future research:
\begin{itemize}
\item Which plane graphs (with specified vertex $\infty$) are the graphs of global flat foldings?
\item What is the computational complexity of recognizing and realizing these graphs?
\item Is there any graph-theoretic difference between the graphs of global flat foldings and the graphs of local flat foldings? In particular does the folding-assignment version of Maekawa's theorem, that each vertex must have two more mountain folds than valley folds or vice versa, impose any nontrivial constraints on the graphs of flat foldings?
\item
\iffull
In Appendix~\ref{sec:orthotree}
\else
In the full version of this paper
\fi
we describe another class of graphs, the \emph{dual orthotrees}, that can always be realized as the graphs of local flat foldings. Can they always be realized as the graphs of global flat foldings?
\item What (if anything) changes when we consider folding patterns on a square sheet of paper (or other bounded shape) rather than on an infinite sheet?
\iffull
Appendix~\ref{sec:outer} begins
\else
In the full version we begin
\fi
a preliminary investigation of this case, in the special case where we restrict the vertex points to the boundary of the paper. On circular paper, all outerplanar graphs are possible, but on square paper, not even all trees can be folded; we find an exact characterization of the foldable trees, different from the characterization in \autoref{sec:trees}. However, similar questions without the restriction to boundary points remain open.
\item Previously we studied algorithms for realizing trees as convex subdivisions of the plane while optimizing the angular resolution of the resulting tree drawing~\cite{CarEpp-GD-06}. Can we use similar ideas to optimize the angular resolution of a folding pattern realization of a tree?
\end{itemize}

\iffull
\bibliographystyle{amsplainurl}
\else
\bibliographystyle{splncs}
\fi
\bibliography{folding}

\newpage
\appendix

\section{Dual orthotrees}
\label{sec:orthotree}

\begin{figure}[t]
\centering\includegraphics[width=0.5\textwidth]{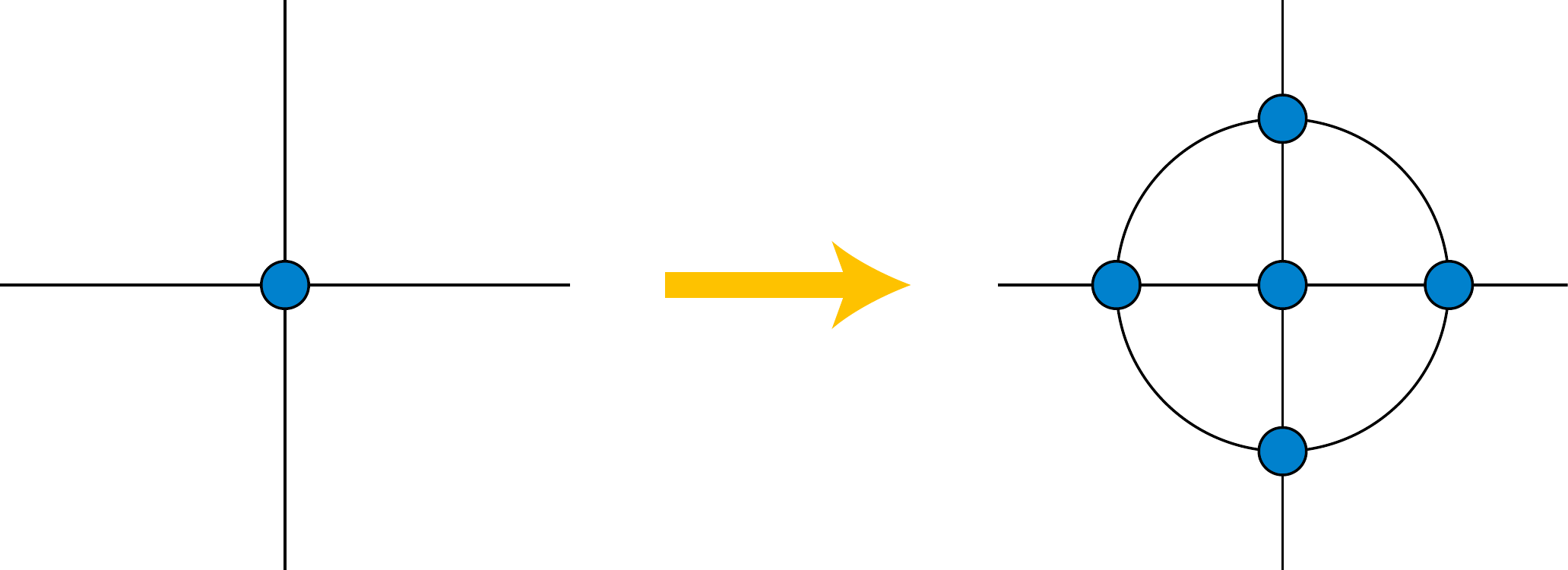}
\caption{The wheel replacement operation}
\label{fig:wheel-replacement}
\end{figure}

We define an operation on plane graphs (such as the graphs of flat foldings), that we call \emph{wheel replacement}. A \emph{wheel} is a planar graph consisting of a cycle and one additional vertex, adjacent to all the cycle vertices. In a wheel replacement operation, we replace one vertex $v$ of a given plane graph, of degree $d$, by a wheel whose cycle has $d$ vertices. We replace each edge of the given graph that is incident to $v$ by an edge incident to one of the cycle vertices of the wheel, in such a way that each cycle vertex has one neighbor outside the wheel and such that the cyclic ordering of these neighbors around the wheel is the same as the cyclic ordering of edges in the original graph incident to $v$. This operation is illustrated in \autoref{fig:wheel-replacement}.

\begin{lemma}
\label{lem:wheel-replacement}
If $G$ is the graph of a local flat folding $\varphi$, then the graph $G_v$ obtained by performing a wheel replacement on any vertex $v\ne\infty$ of $G$ is also the graph of a local flat folding.
\end{lemma}

\begin{proof}
In the local flat folding, the image of a sufficiently small neighborhood of $v$ lies within a wedge of opening angle less than $\pi$. Choose a line that crosses this wedge near $v$, and reflect across this line the points of the neighborhood of $v$ that lie on the same side of this line as $v$. The result of this reflection is another flat folding in which the creases caused by the new reflection form a cycle around $v$, realizing the wheel replacement operation.
\end{proof}

The same realization of a wheel replacement operation can also be visualized as folding over the corner in the sheet of paper formed at~$v$. However, this folding operation cannot always be performed in three-dimensional global flat foldings. The reason is that there might be a crease, disjoint from $v$ in the folding pattern but passing through $v$ in the folded state, that blocks $v$ from being folded over. An alternative 3d realization of the same folding pattern is the \emph{sink folding}, in which the corner is dented inwards (see \cite[p.~33]{Lan-ODS-12}), but again, other nearby parts of the paper may block this fold from being realized.

The graphs that can be constructed from repeated wheel replacement, starting from a multigraph with two vertices and four non-loop edges, include the dual graphs of the surface quadrangulations of \emph{orthotrees}~\cite{DamFlaMei-FWCG-05}, polycubes formed by gluing cubes together in $\R^3$ so that the gluing pattern of the cubes forms a tree. For instance, the \emph{Dal\'\i{} cross}, an unfolded net of a four-dimensional hypercube made famous by Salvador Dal\'\i's painting \emph{Crucifixion (Corpus Hypercubus)}, is an orthotree, and the dual graph of its surface quadrangulation is shown in \autoref{fig:dali-cross}. For this reason we call the graphs formed by repeated wheel replacement starting from the two-vertex four-edge multigraph the \emph{dual orthotrees}.

\begin{figure}[t]
\centering\includegraphics[width=0.8\textwidth]{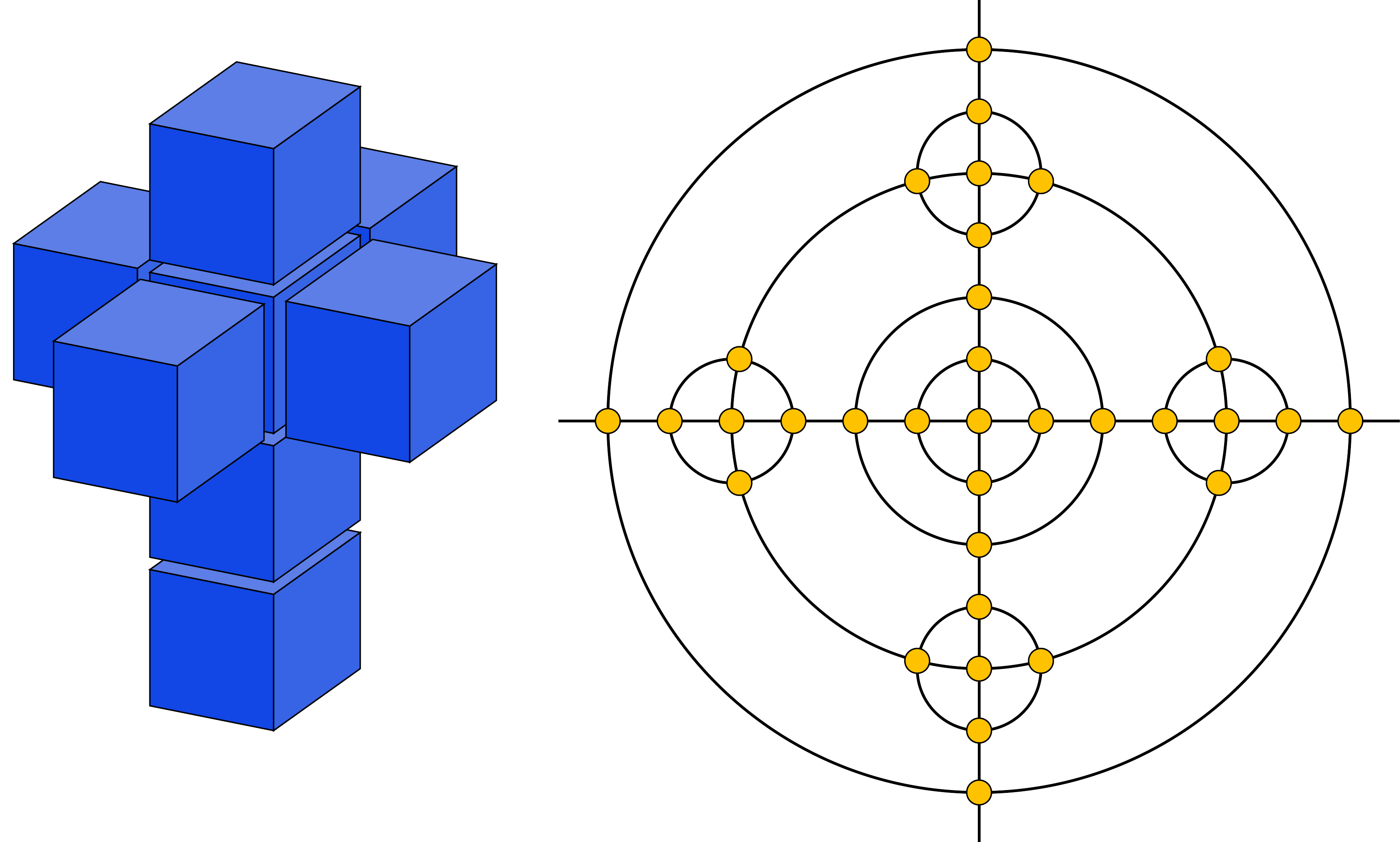}
\caption{The Dal\'\i{} cross (exploded view) and the dual graph of its surface quadrangulation (with one vertex $\infty$ of the dual graph, representing the uppermost square face of the cross, not shown).}
\label{fig:dali-cross}
\end{figure}.

\begin{theorem}
Every dual orthotree is realizable as the graph of a local flat folding.
\end{theorem}

\begin{proof}
By \autoref{lem:wheel-replacement} we can realize every wheel replacement except possibly the ones at the vertex $\infty$. The same folding construction of the lemma can also be performed at $\infty$ only when the image of the flat folding mapping, $\varphi(\R^2)$, lies within a wedge of opening angle less than $\pi$ (as it does for some but not all local flat foldings), rather than covering the entire plane. This property of the flat folding, that the image of $\varphi$ lies within a wedge, is true for the initial fold with one vertex point, and it remains true after each wheel replacement operation at a finite vertex. Therefore, we can perform wheel replacements at all vertices including $\infty$, and by repeated wheel replacements construct any dual orthotree.
\end{proof}

It is possible to realize a wheel replacement operation in a folding pattern by more complicated folds that do not come from a single reflection across a line. There appear to be enough degrees of freedom in the possible realizations of a wheel replacement to allow the realization of any dual orthotree, using an inductive construction in which at each step we ensure that all vertex points
form corners that are unobstructed by other creases. However, we have not found a mathematical description of these realizations for which we can prove that the inductive steps of this construction are always possible. We leave the question of whether every dual orthotree is the graph of a global flat folding as open for future research.

\section{Bounded shapes with boundary vertex points}
\label{sec:outer}

In the earlier parts of this paper we have made the simplifying assumption that the sheet of paper we are folding covers the entire plane. Here, we remove that assumption, and instead study what happens when we use bounded sheets of paper, such as the squares traditionally used for origami. However, we make a different simplifying assumption: that the vertex points of the folding lie on the boundary of the paper.

\subsection{Additional definitions}
We define an \emph{outer local flat folding}, for a given convex region $K$ of the plane,
to be a mapping $\varphi:K\mapsto\R^2$ with the same properties as a local flat folding of the infinite plane: every point of $K$ must be an unfolded point, crease point, or vertex point. However, we additionally require that every vertex point be on the boundary of $K$. Thus, the creases of the folding are \emph{chords} or $R$: line segments that connect two boundary points of $K$, and otherwise pass through the interior of $K$. We define an \emph{outer global flat folding}, as in the case of unbounded sheets of paper, as a local flat folding that can be $\epsilon$-approximated by the vertical projections of three-dimensional topological embeddings of $K$.
We define the \emph{graph} of an outer local or global flat folding to have as its vertices the folded points on the boundary of $K$ (regardless of whether these points are vertex points or crease points) and to have as its edges the pairs of these points that are connected by creases of the folding.

Any non-crossing pattern of finitely many creases on $K$ will describe a valid outer local flat folding. This folding can be constructed by adding one crease at a time, for each new crease composing the mapping function $\varphi$ with the mapping that reflects the plane across the new crease.
We will prove that, when $K$ is a disk or a square, every outer local flat folding is also an outer global flat folding. However, the example in \autoref{fig:unfoldable-triangle} shows that this result does not generalize to other convex shapes such as an equilateral triangle. The folding pattern in the diagram represents a local flat folding that cannot be realized as a global flat folding. The figure has three-way rotational symmetry, with three big triangular flaps surrounding a central equilateral triangle, which is slightly twisted from the outer triangle. Each of the three flaps has a crease separating its sharp corner from the central triangle, with a ``shoulder'' where the crease meets the side of the triangle near a vertex of the central triangle.
If the figure could fold flat globally, two of the three flaps would be on the same side of the central  triangle. But when this happens, the two flaps get in each other's way, so that they can't both be folded flat. If the more clockwise of the two flaps were folded closer to the central triangle, its  shoulder would lie across the crease of the other flap, blocking it from folding. And if the counterclockwise flap were folded closer to the central triangle, its sharp tip would (after folding the crease separating the tip from the central triangle) again lie across the crease of the other flap, blocking it from folding. So no global flat folding is possible.

\begin{figure}[t]
\centering\includegraphics[width=0.4\textwidth]{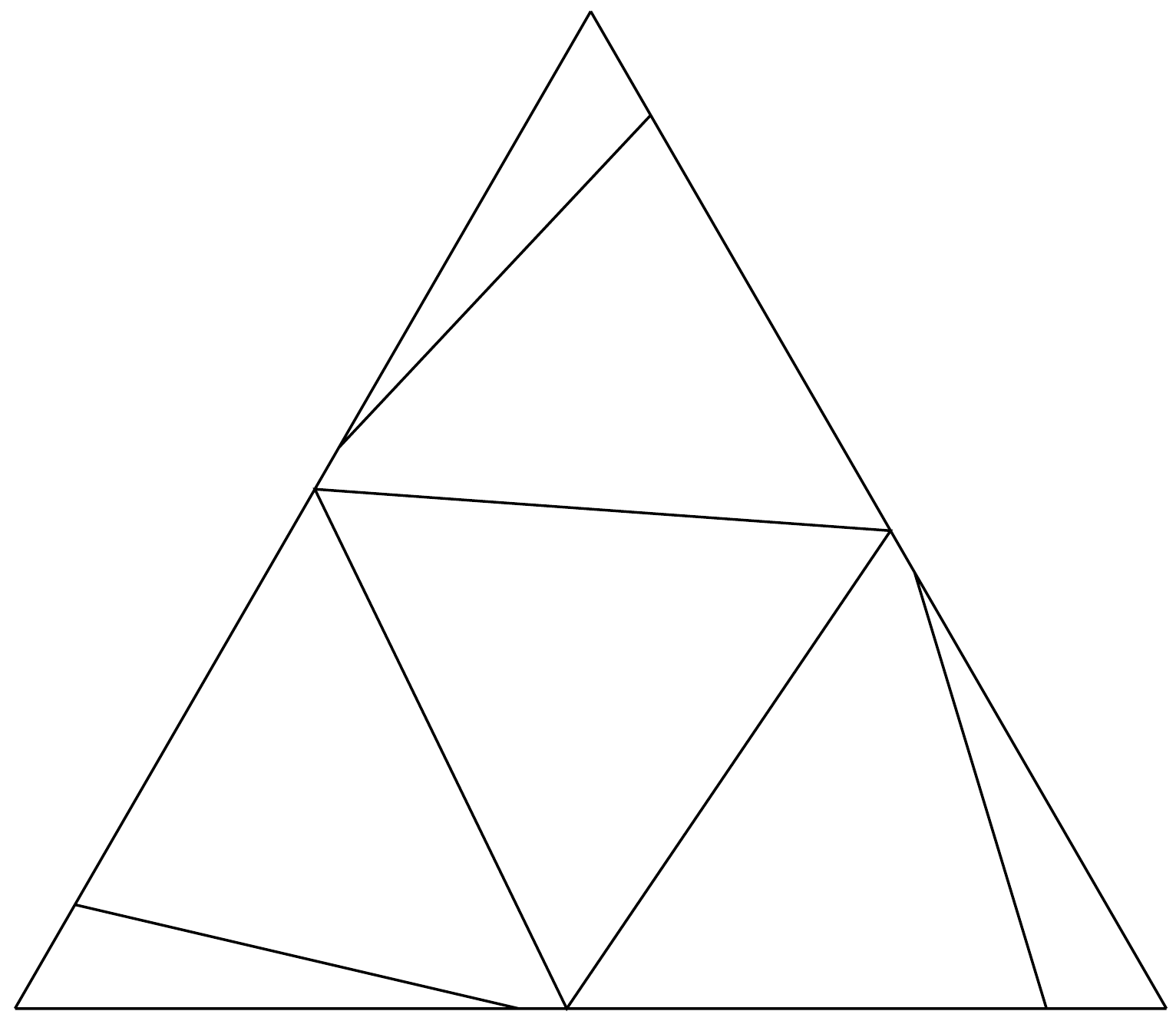}
\caption{A folding pattern for an outer local flat folding of an equilateral triangle that cannot be realized as a global flat folding.}
\label{fig:unfoldable-triangle}
\end{figure}

\subsection{Safe creases}
The following two  lemmas will be very helpful for us in proving that certain outer local flat foldings on certain shapes can also be realized as outer global flat foldings.

\begin{lemma}
\label{lem:inescapable}
Let $u$ and $v$ be boundary folding points of an outer local flat folding $\varphi$ on a given region $K$, such that $\varphi$ includes a crease on line segment $uv$. Suppose also that one of the two regions into which $uv$ partitions $K$, region $C$, has the property that along the boundary curve of $C$ from $u$ to $v$, the distances from $u$ are monotonically increasing and the distances from $v$ are monotonically decreasing. Transform $\varphi$ by a congruence of the plane (if necessary) so that it is the identity mapping on segment $uv$. Then $\varphi(C)$ lies within the union of $C$ and its reflection across $uv$.
\end{lemma}

\begin{proof}
Let $p$ be any point in $C$, and let $q$ be the point on the boundary of $C$ such that $pp'$ is perpendicular to $uv$. Let $L$ be the lune formed by intersecting two disks, centered at $u$ and $v$, with $q$ on their boundary (\autoref{fig:inescapable}). Then $\varphi$ maps $u$ and $v$ to themselves, and cannot increase the distance of any other point from $u$ or from $v$.  Therefore it must map each of the two disks defining $L$ into itself, and (because $L$ is the set of points in both disks) must map $L$ to itself. But $L$ contains $p$ and is entirely contained in the union of $C$ and its reflection, so the image of $p$ must lie within this union.
\end{proof}

\begin{figure}[t]
\centering\includegraphics[width=0.4\textwidth]{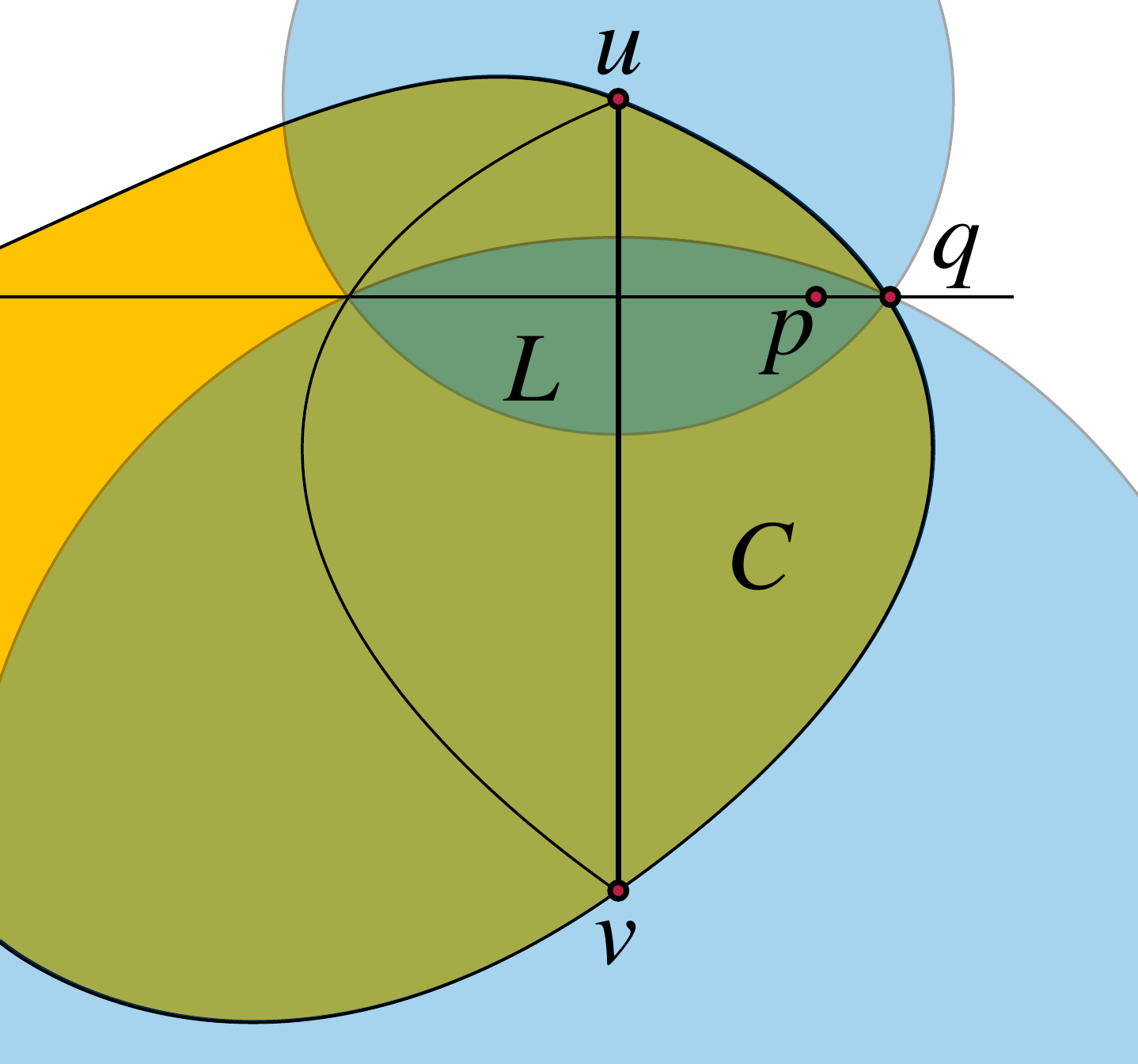}
\caption{Illustration for \autoref{lem:inescapable}}
\label{fig:inescapable}
\end{figure}

\begin{lemma}
\label{lem:safe-crease}
Let $u$ and $v$ be boundary folding points of an outer local flat folding $\varphi$ on a given region $K$, such that $\varphi$ includes a crease on line segment $uv$, splitting off a subregion $C$ of $K$ that meets the conditions of \autoref{lem:inescapable}. Suppose in addition that the reflection across $uv$ of the boundary of $C$ does not cross any crease of $\varphi$.
Let $\varphi_1$ be the outer local flat folding defined by crease $uv$ and all of the creases within the region $C$, and let $\varphi_2$ be the outer local flat folding defined by the remaining creases. Then if both $\varphi_1$ and $\varphi_2$ are outer global flat foldings, then so is $\varphi$.
\end{lemma}

\begin{proof}
Intuitively, we fold $\varphi_1$ first, and then treat its folded state as part of the flat sheet of paper while folding $\varphi_2$. In terms of the $\epsilon$-approximate embeddings that we use to define global flat foldings, what this means is that we construct an embedding of $K$ into $\R^3$ that represents $\varphi_1$, compose it with a transformation that flattens the vertical ($z$) dimension to an arbitrarily small value (so that the image of the embedding is arbitrarily close to the plane again), and then compose this flattened embedding for $\varphi_1$ with the embedding for $\varphi_2$. That is, we apply the mapping of $\varphi(2)$ to the $x$ and $y$ coordinates of the flattened image under $\varphi_1$, and then add the $z$-coordinate of the image under $\varphi_1$ to the result.
\end{proof}

We call a crease $uv$ that meets the conditions of \autoref{lem:safe-crease} a \emph{safe crease}.

\subsection{Disks}

In this subsection we consider the case that the sheet of paper to be folded has the shape of a circular disk. In this case, we can always use safe creases to find global foldings of outer local flat foldings.

\begin{lemma}
\label{lem:safe-chord-exists}
Let $K$ be a disk, let $\varphi$ be a local outer flat folding of $K$, and let $R$ be a region of $K$ bounded by three or more creases of $\varphi$. Then at least one of the bounding creases of $R$ is a safe crease.
\end{lemma}

\begin{proof}
Let $uv$ be the crease bounding $R$ that subtends the smallest angle $\theta$ as viewed from the center point of the disk. Then $\theta<\pi$ so $uv$ meets the conditions of \autoref{fig:inescapable}.
The disk boundary meets line segment $uv$ at angles of $\theta/2$. Any other crease bounding $R$ must subtend an angle from the disk center (on the side containing $uv$) of at least $2\theta$, and therefore it must form an angle at the disk boundary of at least $\theta$. Therefore, the reflection of the disk boundary across $uv$ cannot cross the other disk, and $uv$ is a safe crease.
\end{proof}

\begin{lemma}
\label{lem:disk-local-global}
Every local outer flat folding of a disk is a global outer flat folding.
\end{lemma}

\begin{proof}
We use induction on the number of creases of the folding.
If any region of the disk bounded by creases of the folding is bounded by three or more creases,
we can apply \autoref{lem:safe-chord-exists} to prove that a safe crease exists,
apply the induction hypothesis to the two subsets of creases on either side of the safe crease,
and conclude from \autoref{lem:safe-crease} that the same crease pattern can be realized as a global flat folding. As a base case, if every region of the disk is bounded by only one or two creases, we can pleat the remaining creases to construct a realization as a global flat folding.
\end{proof}

This gives us a complete characterization of the graphs of outer flat foldings of disks:

\begin{theorem}
A graph $G$ can be represented as the graph of a global outer flat folding of a disk if and only if $G$ is outerplanar.
\end{theorem}

\begin{proof}
Place the vertices of $G$ on the boundary of a given disk, in the cyclic order given by their ordering along the outer face of $G$ (skipping repeated copies of the same vertex).
Draw $G$ using straight-line edges, and interpret the resulting drawing as the crease pattern of a local flat folding. By \autoref{lem:disk-local-global} it is also the crease pattern of a global flat folding.
\end{proof}

\begin{figure}[t]
\centering\includegraphics[width=0.25\textwidth]{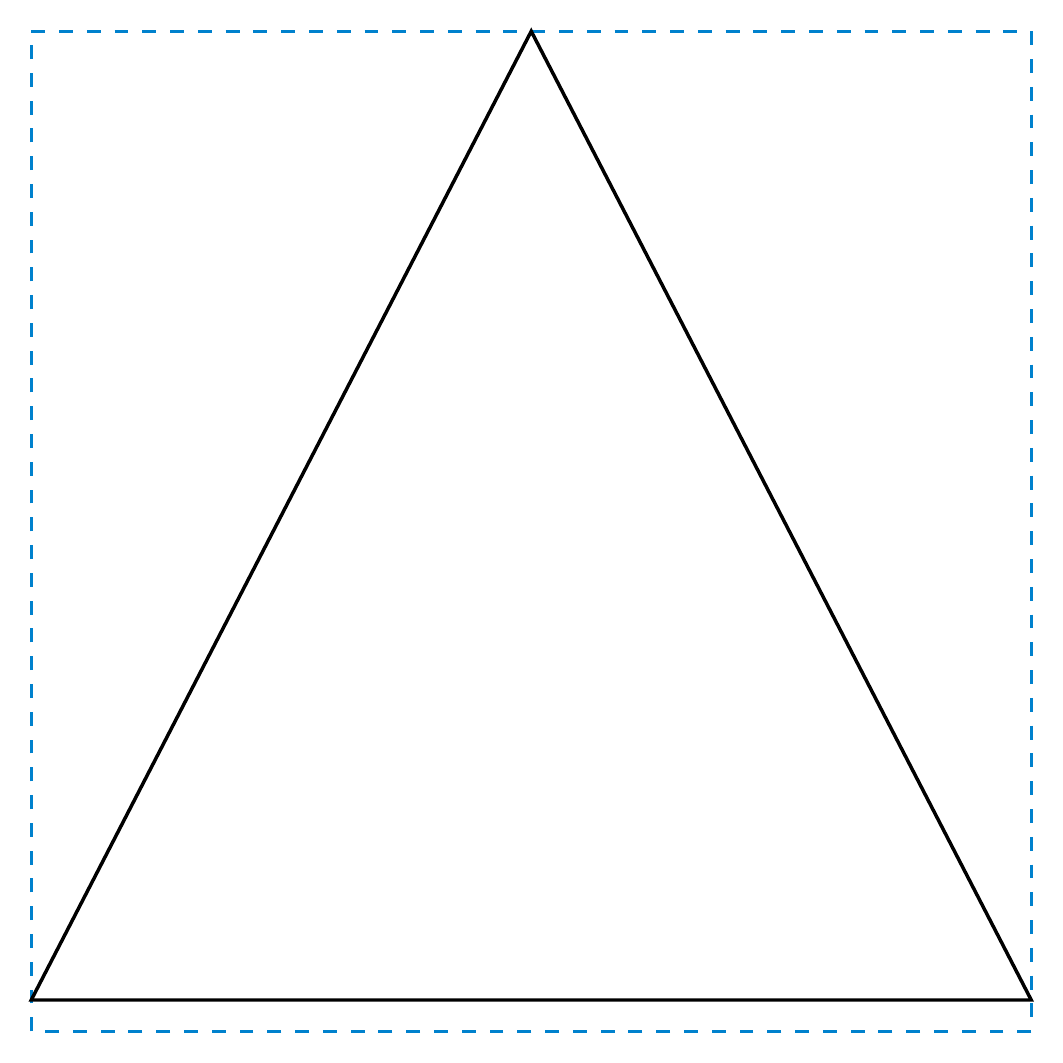}
\caption{A local flat folding of a square in which the region bounded by three creases has no safe crease. The bottom crease does not meet the conditions of \autoref{lem:inescapable}, and the reflection of the square's boundary across either top crease crosses the other top crease.}
\label{fig:no-safe-crease}
\end{figure}

\subsection{Squares}
In contrast to disks,
when $K$ is a square, there may be regions bounded by three or more
creases that have no safe crease (\autoref{fig:no-safe-crease}).
Nevertheless, we use similar concepts to safe creases to prove that outer local flat foldings of a square may be made global.

\begin{lemma}
\label{lem:square-local-global}
Let $\varphi$ be an outer local flat folding of a square $K$.
Then $\varphi$ may be realized as an outer global flat folding.
\end{lemma}

\begin{proof}
Orient $K$ aligned with the coordinate axes of the Cartesian plane.
We may classify the creases of $\varphi$ into six types, according to which pair of distinct sides of $K$ they connect. However, a crease from the top side to the bottom would cross a crease from the left side to the right, so only one of these two types of crease can be present. Without loss of generality (by rotating $K$ if necessary) we may assume that there are no top-to-bottom creases.

If $\varphi$ has creases only in one of the two top corners of $K$ (that is, connecting the top side of $K$ to only one of the left or right sides of $K$), and similarly it has creases only in one of the bottom two corners of $K$) then the folds of $\varphi$ would form a linear sequence that we could safely pleat. And if $\varphi$ has creases in both corners on top (or symmetrically on the bottom) but one of the two of these creases that are farthest from their corners is safe, then we could begin our folding by pleating the creases in the safe corner, eliminating the creases there and reducing to the case where only one of the two top corners has creases. However, as \autoref{fig:no-safe-crease} shows, it may be the case that neither of the two two corners has a safe crease.

\begin{figure}[t]
\centering\includegraphics[width=0.8\textwidth]{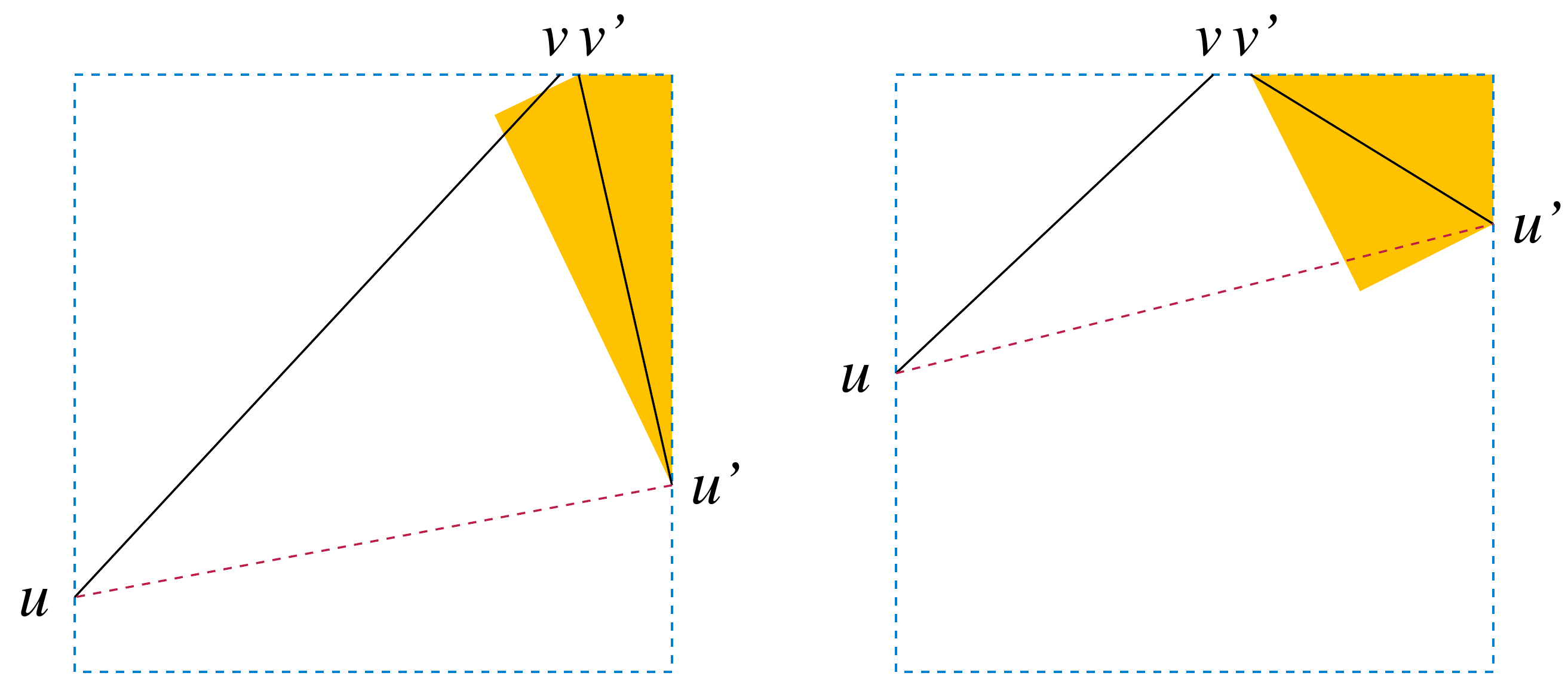}
\caption{The two cases for identifying $u'v'$ in the upper right corner of square $K$ as a semi-safe crease. The reflected image of the upper right corner across crease $u'v'$ can interfere with crease $uv$ in the upper left corner (left image), or with creases below line $uu''$ (bottom) but not both.}
\label{fig:semi-safe}
\end{figure}

In this case, let $uv$ and $u'v'$ be the two creases on the top left and top right, respectively, that are farthest from their corners, with $u$ on the left side of $K$, $v$ and $v'$ on the top side of $K$, and $u'$ on the right side of $K$. Assume without loss of generality that $u'$ has at least as large a $y$-coordinate as $u$ (otherwise flip $K$ from left to right to achieve this, without affecting its global foldability). Then it is possible for the fold at $u'v'$ to interfere with the fold at $uv$ (if the reflected image of the top right corner of the square across $u'v'$ crosses $uv$), or for it to interfere with folds below line $uu'$ (again, if the reflected image of the top right corner of the square crosses this line)
but only one of these two types of interference can happen. For, the fold at $u'v'$ can interfere with the fold at $uv$ only if $v'$ is closer to the top right corner of $K$ than $u'$, so that the reflected image of the corner has positive slope at $v'$ (\autoref{fig:semi-safe}, left). But the fold at $u'v'$ can interfere with folds below $uu'$ only if $u'$ is closer to the top right corner of $K$ than $v'$, so that the reflected image of the corner has positive slope at $u'$ (\autoref{fig:semi-safe}, right). Only one of these two things can happen. Because of this, we call crease $u'v'$ \emph{semi-safe}.

In the same way, if there are creases in both of the bottom two corners, we can identify one of the two creases farthest from their corner as being semi-safe. If we remove from the folding pattern of $\varphi$ these semi-safe creases and all of the other creases in the same corner, then the remaining creases (in the other two corners, and from one side of the square to the other) form a linear sequence that can be pleated. Before we make this pleat of the remaining creases, however, we will pleat the creases in each semi-safe corner of $K$. There are two choices for how to perform the pleat in each semi-safe corner (starting first with a mountain fold, or with a valley fold) and we make these choices according to the following case analysis:
\begin{itemize}
\item If there are no semi-safe corners, then all the creases of $\varphi$ can be pleated.
\item If there is a single semi-safe corner, then we pleat the creases in that corner
starting with whichever of a mountain fold or valley fold is opposite to the closest crease in the remaining folds that it interferes with (choosing arbitrarily if it doesn't interfere with any folds).
In this way, the folded semi-safe corner is placed between two sheets of the 3d pleat of the remaining creases such that the crease where these two sheets meet is one that it does not interfere with.
\item If there are two semi-safe corners that are separated from each other by at least one crease of the remaining creases of $\varphi$, we handle each one independently in the same way that we handled the case where there is only one semi-safe corner. We do the same if the top semi-safe corner interferes with the crease in the other top corner, and the bottom semi-safe corner interferes with the crease in the other bottom corner, because then these two semi-safe corners cannot interfere with each other.
\item If there is no crease separating the two semi-safe corners, the top semi-safe corner interferes with some creases that are not in the top corners, but the bottom semi-safe corner interferes only with creases in the other bottom corner, then we fold the bottom semi-safe corner first, with the opposite starting fold to the other bottom corner. This initial fold cannot interfere with any folds in the top corners, and its starting fold orientation is chosen in such a way that it also cannot interfere with any folds in the other bottom corner. Once we have made this fold, we can fold the top semi-safe corner, again using the opposite starting fold to the other bottom corner. This folding of the top semi-safe corner cannot interfere with the bottom semi-safe corner (because we have already folded it) nor with the other bottom corner (because it starts with a fold of the opposite orientation). Finally we pleat the remaining folds. The case when the top semi-safe corner interferes with creases in the top corner and the bottom semi-safe corner interferes with creases in the top corners is symmetric.
\item In the remaining case, there is no crease separating the two semi-safe corners, the bottom semi-safe corner interferes with creases in the top corners, and the top semi-safe corner interferes with creases in the bottom corners. (In particular the two semi-safe corners could interfere with each other, so they must be given opposite starting orientations). In this case it follows from the lack of a fold separating the corners that, in the pleat of the remaining folds outside of the semi-safe corners, the nearest fold in the non-semi-safe top and bottom corners (if they exist) have opposite orientations. We give the top semi-safe fold the same orientation (mountain or valley) as the neighboring fold in the other top corner, and the bottom semi-safe fold the same orientation as the neighboring fold in the other bottom corner. In this way, the two semi-safe folds can neither interfere with each other (as they have opposite starting fold orientations) nor with any other of the remaining folds. 
\end{itemize}
\end{proof}

The graphs that can be graphs of outer flat foldings of the square are not as easy to characterize as for the disk.  To simplify their description, we limit our attention to trees.
We define the \emph{spine} of a tree of three or more vertices to be the subtree formed by removing all degree-one vertices; for instance, the caterpillars are the trees whose spine is a path.

\begin{lemma}
\label{lem:few-spine-leaves}
Let $T$ be a tree that is the graph of an outer local flat folding of a convex $k$-gon.
Then the spine of $T$ has at most $k$ leaves.
\end{lemma}

\begin{proof}
Perturb the folding points on the boundary of the folding, if necessary, so that no folding point on the boundary of the $k$-gon lies at one of its corners.
For each leaf vertex $v$ of the spine, choose a degree-one neighbor $w$ of $v$ in $T$; $w$ must exist or else $v$ would either have been removed from the spine or would have a child in the spine.
Then crease $vw$ separates at least one vertex of the $k$-gon from the spine of $T$.
No other leaf vertex $v'$ of the spine can have a child crease $v'w'$ that separates the same vertex of the $k$-gon from the spine, because $v'$ is on the spine side of $vw$ so $v'w'$ does not separate the vertex of the $k$-gon from $v$. Therefore, the number of leaf edges of the spine is at least the number of vertices of the $k$-gon, which is $k$.
\end{proof}

\begin{lemma}
\label{lem:spine-fold}
Let $K$ be a $k$-gon and let $T$ be a tree whose spine has at most $k$ leaves.
Then $T$ is the graph of a local outer flat folding for $K$.
\end{lemma}

\begin{figure}[t]
\centering\includegraphics[width=0.9\textwidth]{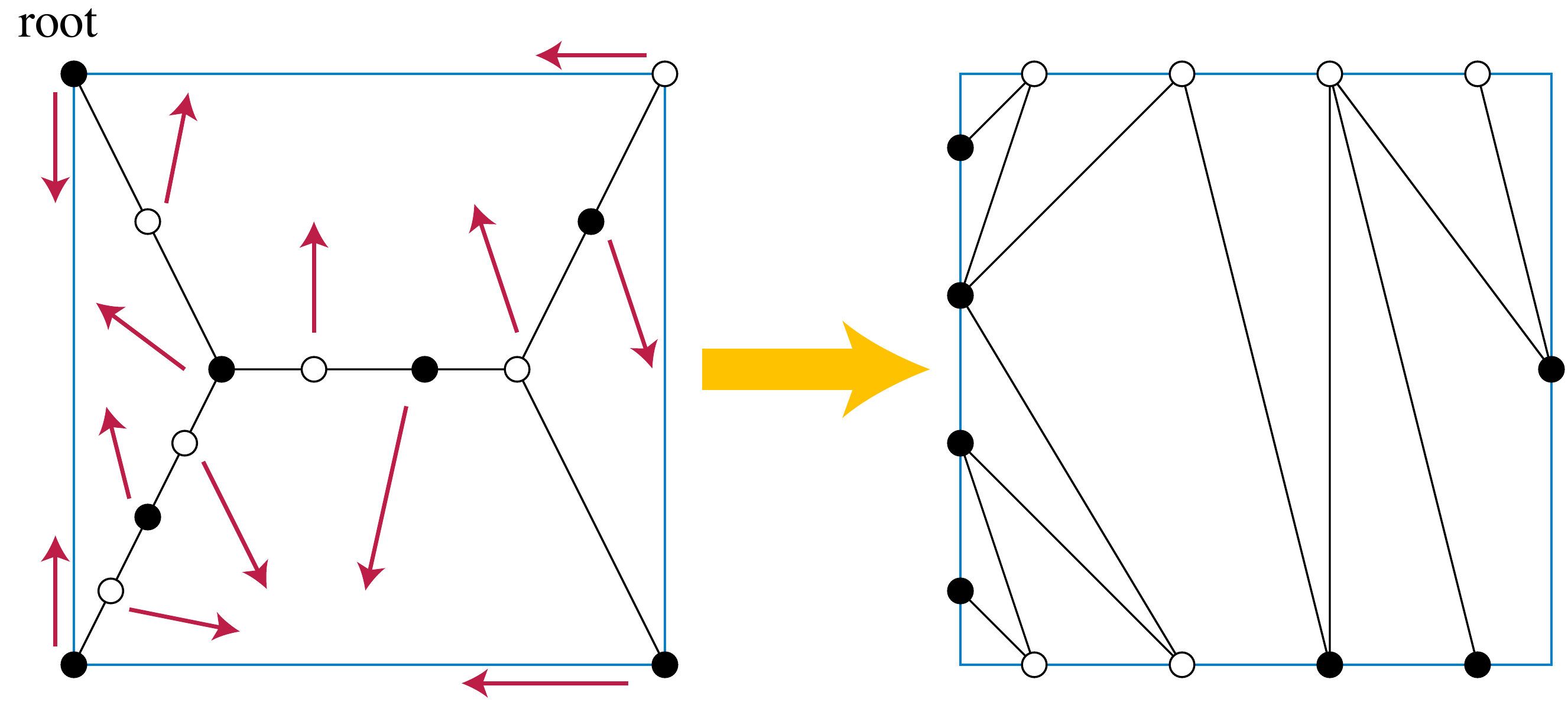}
\caption{Folding the spine of a tree $T$ onto a polygon with as many sides as the spine has leaves (\autoref{lem:spine-fold}).}
\label{fig:spine-fold}
\end{figure}

\begin{proof}
Draw the spine of $T$ on $K$ so that its leaves are at the vertices of $K$ and its other edges are interior to $K$. Each face of this drawing is bounded by at least one side of $K$; choose one of these sides as the label for its face. Choose one of the leaves of $T$ as its root, and two-color $T$ by distance from the root, black at even distances and white at odd distances (\autoref{fig:spine-fold}, left). Shift each vertex of the drawing onto the boundary of $K$, shifting the black vertices onto the side of $K$ given by the label of their leftmost incident face and the white vertices onto the side given by the label of their rightmost incident face, and preserving the order of the vertices mapped to the same side of $K$ within each region. (\autoref{fig:spine-fold}; the red arrows on the left show the shifting direction and the right side of the figure shows the result.)

Then in the resulting drawing of the spine, each spine vertex has a visible segment of polygon boundary that it is not on: for a black vertex this is the side of $K$ given by the label of its rightmost incident face (the one it didn't shift onto) and for a white vertex this is the side of $K$ given by the label of its leftmost incident face. All of the leaf vertices of $T$ can be drawn by connecting spine vertices to folding points on these visible segments of polygon boundary.
\end{proof}

This completes our characterization of trees that can be the graphs of outer foldings on a square:

\begin{theorem}
A tree $T$ is the graph of an outer global flat folding on a square if and only if the spine of $T$ has at most four leaves.
\end{theorem}

\begin{proof}
The impossibility of realizing trees with more spine leaves is \autoref{lem:few-spine-leaves}.
If $T$ does have four or fewer spine leaves, we can find a realization as an outer local flat folding by \autoref{lem:spine-fold} and convert it to an outer global flat folding by \autoref{lem:square-local-global}.
\end{proof}

\end{document}